\documentclass[submission,copyright,creativecommons]{eptcs}
\providecommand{\event}{TARK 2023} \providecommand{\keywords}[1]{\textbf{\textit{Keywords---}} #1}

\usepackage{iftex}
\usepackage{graphicx}
\usepackage{parskip}
\usepackage{indentfirst}
\usepackage{enumitem}
\usepackage{comment}
\usepackage{lscape}
\usepackage{booktabs,caption}
\usepackage[flushleft]{threeparttable}
\usepackage[english]{babel}
\usepackage{xcolor}
\usepackage{lipsum}
\usepackage{float}
\usepackage{amssymb}
\usepackage{amsmath}
\usepackage{amsthm}
\usepackage{algorithm} 
\usepackage{algpseudocode} 
\usepackage[autostyle]{csquotes}  
\usepackage[numbers]{natbib}

\newtheorem{definition}{Definition}[section]
\newtheorem{proposition}{Proposition}[section]
\newtheorem{theorem}{Theorem}[section]
\newtheorem{corollary}{Corollary}[section]

\newtheorem{example}{Example}[section]

\ifpdf
  \usepackage{underscore}        
  \usepackage[T1]{fontenc}        
\else
  \usepackage{breakurl}           
\fi

\title{Cognitive Bias and Belief Revision}
\author{Panagiotis Papadamos
\institute{Technical University of Denmark}
\email{panagiotispapadamos@gmail.com }
\and
Nina Gierasimczuk \qquad
\institute{Technical University of Denmark}
\email{\quad nigi@dtu.dk}
}
\def\titlerunning{Cognitive Bias and Belief Revision}
\def\authorrunning{P. Papadamos, N. Gierasimczuk}
\begin{document}
\maketitle

\begin{abstract}
In this paper we formalise three types of cognitive bias within the framework of belief revision: confirmation bias, framing bias, and anchoring bias. We interpret them generally, as restrictions on the process of iterated revision, and we apply them to three well-known belief revision methods: conditioning, lexicographic revision, and minimal revision. We investigate the reliability of biased belief revision methods in truth-tracking. We also run computer simulations to assess the performance of biased belief revision in random scenarios. 
\end{abstract}


\section{Introduction}
Cognitive bias is a systematic human thought pattern connected with the distortion of received information, that usually leads to deviation from rationality (for a recent analysis see \cite{fast-and-slow-thinking}). Such biases are specific not only to human intelligence, they can be also ascribed to artificial agents, algorithms and programs. For instance, confirmation bias can be seen as stubbornness  against new information which contradicts the previously adopted view. In some cases such confirmation bias can be implemented into a system purposefully. Take as an example an authentication algorithm and a malicious user who is trying to break into an email account. Say that the algorithm, before it locks the access, allows only three attempts to enter the correct password. Hence, the algorithm (temporarily) insists that the user who tries to connect is the real holder of the credentials, despite the input being inconsistent with that hypothesis. The algorithm will not revise its `belief' about the user's identity, until it receives the evidence to the contrary a specific number of times. Another unorthodox example of a biased artificial agent concerns anchoring bias, where an agent makes a decision based on a recent, selected piece of information, possibly ignoring other data. In the context of artificial agents, such situations may occur justifiably when resources (like time or memory) are limited. As an example consider two computers, $A$ and $B$, connected within a network. Computer $A$ attempts to communicate with computer $B$, but for some reason, computer $A$ does not receive $B$'s response within a specified time range and, as a result, erroneously considers $B$ dead. This inability to communicate leads computer $A$ to change its `belief' about $B$'s liveness, and, subsequently, to make decisions based on this distortion.

In this paper we study some dynamic aspects of three types of cognitive bias: confirmation bias, framing bias, and anchoring bias. We will apply them to three well-known belief revision methods: conditioning, lexicographic, and minimal revision \cite{Spohn:1988aa,Rott:1989ab,Boutilier:1993fj,Benthem:2007aa}. We first recall the background of the model of truth-tracking by belief revision from \cite{Gie10,BGS11,Baltag:2019ab} (related to earlier work in \cite{Kel98,Kelly:1999aa}, see also \cite{Gierasimczuk:2013aa}), which borrows from computational learning theory, and identifiability in the limit in particular \cite{Gold67,JORS99}. We proceed by investigating the effect of bias on truth-tracking properties of various belief revision policies. Finally, we present our computer simulation in which we empirically compare the performance of biased and regular belief revision in different scenarios. We close with several directions of further work.

\subsection{Background: truth-tracking and belief revision}
We will now introduce basic notions, following the framework of truth-tracking by belief revision proposed in \cite{Baltag:2019ab}.  
Our agents' uncertainty space will be represented by a so-called \emph{epistemic space}, $\mathbb{S}=(S,\mathcal{O})$, where $S$ is a non-empty, at most countable set of worlds (or states), and $\mathcal{O}\subseteq \mathcal{P}(S)$ is a set of possible observations. We will call any subset $p$ of $S$ a \emph{proposition}, and we will say that a proposition $p$ is \emph{true in $s\in S$} if $s\in p$. 

Data streams and sequences describe the information an agent receives over time. A \emph{data stream} is an infinite sequence of observations $\vec{O}=(O_{0}, O_{1},\ldots)$, where $O_{i}\in \mathcal{O}$, for $i\in\mathbb{N}.$ A \emph{data sequence} is a finite initial segment of a data stream; we will write $\vec{O}[n]$ for the initial segment of $\vec{O}$ of length $n$, i.e., $\vec{O}_0,\vec{O}_1,\ldots, \vec{O}_{n-1}$. Given a (finite or infinite) data sequence $\sigma$, $\sigma_{n}$ is the $n$-th element of in $\sigma$; $set(\sigma)$ is the set of elements enumerated in $\sigma$; $\#O(\sigma)$ is the frequency of observation $O$ in $\sigma$; let $\tau$ be a finite data sequence, then $\tau\cdot\sigma$ is the concatenation of $\tau$ and $\sigma$. 
A special type of data streams are \emph{sound and complete} streams. A data stream $\vec{O}$ is \emph{sound with respect to a state $s\in S$} if and only if every element in $\vec{O}$ is true in the world $s$, formally $s\in\vec{O}_{n}$, for all $n\in \mathbb{N}$. A data stream $\vec{O}$ is \emph{complete with respect to a state $s\in S$} if and only if every proposition true in $s$ is in $\vec{O}$, formally if $s\in O$ then there is an $n\in\mathbb{N}$, such that $O=\vec{O}_{n}$. Sound and complete streams form the most accommodating conditions for learning.

\begin{definition}
    Given an epistemic space $\mathbb{S}=(S,\mathcal{O})$ and a data sequence $\sigma$, a \emph{learning method $L$} (also referred to it as a \emph{learner}), is a function that takes as an input the epistemic space $\mathbb{S}$ and the sequence $\sigma$, and returns a subset of $S$, $L(\mathbb{S},\sigma)\subseteq S$, called a \emph{conjecture}. 
\end{definition}

The goal of learning is to identify the \emph{actual world}, which is a special designated element of the epistemic space. Given the epistemic space of an agent and the incoming information, which is (to some degree) trusted, the agent learns facts about the actual world step by step in order to achieve its goal, identifying the actual world. 

\begin{definition}
Let $\mathbb{S}=(S,\mathcal{O})$ be an epistemic space, $s\in \mathbb{S}$ is identified in the limit by $L$ on $\vec{O}$, iff there is a $k$, such that for all $n\geq k$, $L(\mathbb{S},\vec{O}[n])=\{s\}$; $s\in \mathbb{S}$ is identified in the limit by $L$ iff $s$ is identified in the limit by $L$ on every sound and complete data stream for $s$; $S$ is identified in the limit by $L$ if all $s\in S$ are identified in the limit on by $L$; Finally, $\mathbb{S}$ is \emph{identifiable in the limit} iff there exists an $L$ that identifies it in the limit.
\end{definition}

\smallskip

To be able to talk about beliefs of our agents (and whether or not they align with the actual world), we add to the epistemic space a plausibility relation. 
Given an epistemic space $\mathbb{S}=(S,\mathcal{O})$, a \emph{prior plausibility} assignment ${\preceq}\subseteq S\times S$ is a total preorder. Such $\mathbb{S}^\preceq =(S, \mathcal{O}, \preceq)$ will be called a plausibility space (generated from $\mathbb{S}$, for simplicity of our notation we will often refer to such space with $\mathbb{B}$). The prior plausibility assignment is not fixed---it may be different for different agents, and serves as starting points of their individual belief revision processes. Plausibility models allow defining beliefs of agents. For any proposition $p$, we will say that the agent believes $p$ in $\mathbb{S}^\preceq$ if $p$ is true in all worlds in $min_\preceq (S)$. 

Plausibility spaces, and hence also beliefs, change during the belief revision process. We will focus on three popular belief revision methods that can drive such a learning: conditioning, lexicographic, and minimal belief revision. 
\begin{definition}[Revision method]\label{cond}\label{lex}\label{min}\label{one-step-methods}  A \emph{one-step revision method $R_1$} is a function such that for any plausibility space $\mathbb{B}=(S,\mathcal{O},\preceq)$ and any observable proposition $p\in \mathcal{O}$ returns a new plausibility space $R_{1}(\mathbb{B},p)$. We define three one-step revision methods:
    \begin{enumerate}
    \item[] \emph{Conditioning}, $Cond_1$, is a one-step revision method that takes as input a plausibility space $\mathbb{B}=(S,\mathcal{O},\preceq)$ and a proposition $p\in \mathcal{O}$ and returns the restriction of $\mathbb{B}$ to $p$. Formally, $Cond(\mathbb{B},p)=(S^{p},\mathcal{O},{\preceq}^{p})$, where $S^{p}=S\cap p$ and ${\preceq}^{p}={\preceq}\cap(S^{p}\times S^{p})$.
    \item[] \emph{Lexicographic revision}, $Lex_1$, is     a one-step revision method that takes as input a plausibility space $\mathbb{B}=(S,\mathcal{O},\preceq)$ and a proposition $p\in \mathcal{O}$ and returns a plausibility space $Lex(\mathbb{B},p)=(S,\mathcal{O},\preceq^{'})$, such that for all $t, w\in S$, $t\preceq^{'}w$ if and only if $t\preceq_{p}w$ or $t\preceq_{\bar{p}}w$ or $(t\in p$ and $w\notin p)$, where ${\preceq_{p}}={\preceq}\cap(p\times p)$, ${\preceq}_{\bar{p}}={\preceq}\cap(\bar{p}\times\bar{p})$, and $\bar{p}$ is the complement of $p$ in $S$.
    \item[] \emph{Minimal revision}, $Mini_1$, is     a one-step revision method that takes as input a plausibility space $\mathbb{B}=(S,\mathcal{O},\preceq)$ and a proposition $p\in \mathcal{O}$ and returns a new plausibility space $Mini(\mathbb{B},p)=(S,\mathcal{O},\preceq^{'})$ where for all $t,w\in S$, if $t\in \min_{p}$ and $w\notin \min_{p}$, then $t\preceq^{'}w$, otherwise $t\preceq^{'}w$ if and only if $t\preceq w$.
      \end{enumerate}

  \noindent An \emph{iterated belief revision method $R$} is obtained by iterating the one-step revision method $R_{1}$: $R(\mathbb{B},\lambda)=\mathbb{B}$ if $\lambda$ is an empty data sequence, and $R(\mathbb{B},\sigma\cdot p)=R_{1}(R(\mathbb{B},\sigma),p)$.
\end{definition}

\begin{definition}\label{def:br_learning}
Let $R$ be an iterated belief revision method, $S^\preceq$ a plausibility space, and $\vec{O} $ a stream. A \emph{belief revision based learning method} is defined in the following way: $L_R^\preceq(\mathbb{S},\vec{O} [n])=min_\preceq R(\mathbb{S}^\preceq,\vec{O} [n])$.

We will say that the revision method $R$ identifies $\mathbb{S}$ in the limit iff there is a $\preceq$ such that $L_R^\preceq$ identifies $\mathbb{S}$ in the limit. A revision method $R$ is \emph{universal} on a class $\mathbb{C}$ of epistemic spaces if it can identify in the limit every epistemic space $\mathbb{S}\in\mathbb{C}$ that is identifiable in the limit. 
\end{definition}

\begin{theorem}[\cite{Baltag:2019ab}]\label{cond_uni}\label{lex_uni}\label{mini_uni}\label{th:br_uni}
    The belief revision methods $Cond$ and $Lex$ are universal, while $Mini$ is not.
\end{theorem}

Learning methods can be compared with respect to their power. We will say that a learner $L'$ is \emph{at least as powerful as} learner $L$, $L\sqsubseteq L'$, if every epistemic space $\mathbb{S}$ that is identified in the limit by $L'$, is identified in the limit by $L$. We will say that $L'$ is \emph{strictly more powerful than} learner $L$, if $L\sqsubseteq L'$ and it's not the case that $L'\sqsubseteq L$. Analogously, using definition \ref{def:br_learning}, we will apply the same terms to belief revision methods.

In the remainder of this paper we will discuss several ways of introducing cognitive bias into this picture of iterated belief revision and long-term truth-tracking, together with computer simulation results that paint a more quantitative picture of the analytical results.

\section{Simulating belief revision}\label{sec:simulating_br} Throughout this work we also present the results of computer simulations we run to see how various (biased) methods compare with respect to their truth-tracking ability. To this end we implemented  artificial belief revision agents (for the biased and unbiased scenarios), which try to identify a selected actual world on the basis of sound and complete streams. We use the object-oriented programming language Python. The code can be found in the repository of the project \cite{repo}, and the structure of the code can be seen in Figure \ref{fig:structure_before_instantiation}. 
\begin{figure}[htbp]
    \centering
    \includegraphics[width=0.5\textwidth]{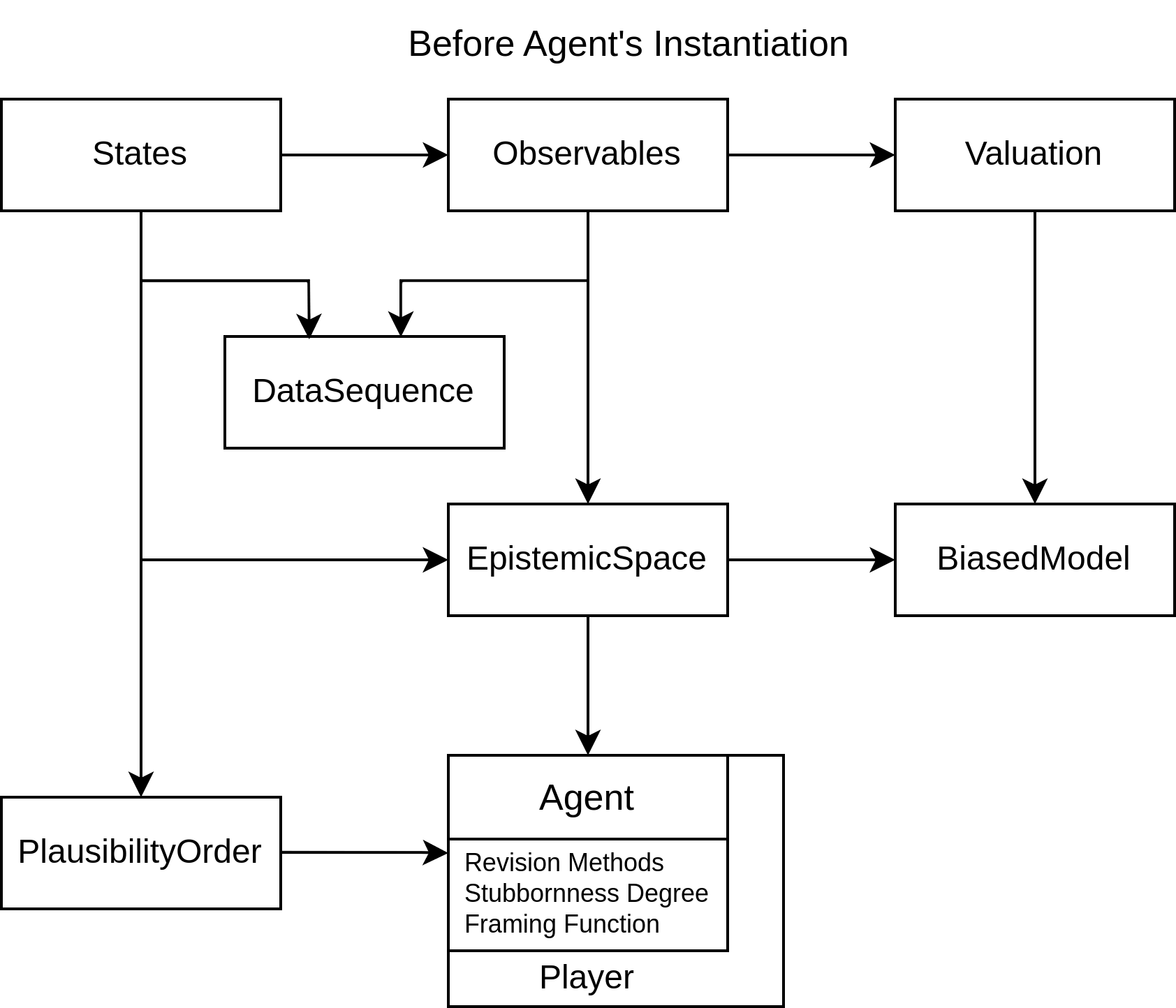}
    \caption{Communication of classes in the implementation}
    \label{fig:structure_before_instantiation}
\end{figure}

The simulation included both custom and random tests. Custom tests were created to check the correctness of the implemented functions, while random tests were created to investigate the reliability and the performance of the (biased) belief revision methods. In the implementation all plausibility spaces are finite. This choice is governed by the practicality of the implementation. We ran several series of tests. Each series of tests consisted of $200$ tests, while the plausibility spaces consisted of $\approx 5$ possible states and $\approx 12$ observables, and the incoming data sequence was longer than the number of observables ($\approx 2-4$ more observables). These numbers were hard-coded to ensure computational feasibility of the experiment. The plausibility spaces we created for the automatic tests were completely random and so could turn out to be unidentifiable. This is the reason why there were identification failures for the universal revision methods, even for unbiased cases. 
After we randomly generated an epistemic space, one of the states (let us call it $s$) was randomly designated to be the actual world, and a sound and complete data stream $\sigma$ for $s$ was generated. A plausibility preorder over the epistemic was then randomly generated (generating a plausibility space).  We then called on each of the (biased) revision methods and made them attempt to identify $s$ from $\sigma$. As we will also see in the later comparisons, overall the frequencies of successful identification by unbiased (regular) belief revision methods were very high across experiments: for conditioning between $94\%$ and $98\%$, for lexicographic revision between $97\%$ and $99\%$, and for minimal revision between $77\%$ and $82\%$.

\section{Cognitive bias and belief revision}\label{biased_model}
We will propose abstract accounts of three types of cognitive bias: confirmation bias, framing bias, and anchoring bias. For each we will describe how an agent revises its belief. We will see how the bias affects truth-tracking, both theoretically, through a learning-theoretic analysis of (non-)universality, and practically, in computer simulations.

\subsection{Confirmation Bias}
Hahn and Harries \cite{Hahn:2014aa} characterized confirmation bias as a list of four `cognitions', namely: hypothesis-determined information seeking, failure to pursue falsification strategy in the context of conditional reasoning, stubbornness  to change of belief once formed, and overconfidence or illusion of validity of our belief. The first cognition will not concern us, as we don't focus on agents that actively seek information, but rather we focus on how passive agents perceive incoming information.  

To analyse selective bias, given a space $\mathbb{S}=(S,\mathcal{O})$, we could designate a subset of $\mathcal{O}$ to be the set of propositions that are `important' to the agent. We would then allow that they are given a special, privileged treatment during the revision process. We choose to express this level of importance more generally with a numerical assignment, which we call the \emph{stubbornness function}.
\begin{definition}\label{stub_degree}
    Given an epistemic space $\mathbb{S}=(S,\mathcal{O})$, the \emph{stubbornness   function} is $D:\mathcal{P}(S)\rightarrow\mathbb{N}$.
\end{definition}

The stubbornness function describes the level of an agent's bias towards a proposition, intuitively the ones with stubbornness degree higher than $1$ can be considered important to the agent. The higher the stubbornness degree, the more biased the agent is towards the proposition, so the more difficult it is to change its belief in that proposition---there should be strong evidence against it. For an unbiased agent the value of the function $D$ for every proposition is $1$. An unbiased agent will revise its beliefs instantly after it receives information inconsistent with its beliefs. An agent that is biased towards a proposition $p$ and believes $p$, should receive information `$\neg p$' $D(p)$-many times in order to react by revising its belief with $\neg p$. The agent struggles with falsifying its belief, maintains the illusion of its belief's validity, by resisting change. 

For each one-step revision method $R_1$ given in Definition \ref{one-step-methods}, we will provide a confirmation-biased version or iterated revision $R_{CB}$. $R_{CB}$ will take a plausibility space and a sequence of data and output a new plausibility space. Intuitively, it will attempt to execute the unbiased version of the revision method, but this will only succeed if the stubbornness degree allows it, i.e., if the data contradicting the proposition is repeated enough times.

\begin{definition}[Confirmation-biased revision methods]\label{conf-biased-methods}
Let $\mathbb{B}=(S,\mathcal{O},\preceq)$ be a plausibility space and let $D$ be a stubbornness   function, $\sigma\in \mathcal{O}^{*}$ be a data sequence\footnote{Let $\Sigma$ be a set, then $\Sigma^\ast$ is a set of all finite sequences of elements from $\Sigma$.}, $p\in \mathcal{O}$ be an observable and $R_{1}$ is a one-step revision method. A confirmation-bias belief revision method $R_{CB}$ is defined in the following way:
\[R_{CB}(\mathbb{B},\lambda)=\mathbb{B},  \]  \[R_{CB}(\mathbb{B},\sigma\cdot p)=
\begin{cases}
      R_{1}(R_{CB}(\mathbb{B},\sigma),p) & \text{if}\quad \#p(\sigma)\geq D(\overline{p}),\\
      R_{CB}(\mathbb{B},\sigma) & \text{otherwise}.
\end{cases} 
\]
where $\lambda$ is an empty sequence, $\#p(\sigma)$ stands for the number of occurrences of $p$ in $\sigma$, and $\overline{p}$ the complement of $p$ in $S$.

We obtain the confirmation-biased conditioning, lexicographic and minimal revision $Cond_{CB}$, $Lex_{CB}$, $Mini_{CB}$ by substituting $R_{1}$ in the preceding definition by $Cond_{1}, Lex_{1}$, and $Mini_{1}$, respectively.
\end{definition}\label{cb_definition}

\paragraph{Truth-tracking under confirmation bias}\label{truth-tracking-confirmation-biased}
An agent under confirmation bias updates its belief with respect to the stubbornness degree. Below we see that it is the crucial factor that breaks the universality of the belief revision methods. 

\begin{proposition}\label{conBiasEx}\label{lexBiasEx}\label{prop:CB_truthtracking}
$Cond$, $Lex$ and, $Mini$ are strictly more powerful than $Cond_{CB}$, $Lex_{CB}$, and $Mini_{CB}$, respectively.
\end{proposition}

\begin{proof}
We will give an example of an epistemic space $\mathbb{S}=(S,\mathcal{O})$ that is identified by $Cond$, but is not identified by $Cond_{CB}$. Let $\mathbb{S}=(S,\mathcal{O})$, where $S=\{w,t,s,r\}$, $\mathcal{O}=\{p,q,\bar{p},\bar{q}\}$ and $p=\{w,t\}$,$\bar{p}=\{s,r\}$, $q=\{w,s\}$, and $\bar{q}=\{t,r\}$. Clearly, this space is identifiable by regular conditioning method $Cond$: take the plausibility order that takes all worlds to be equally plausible. Then, whichever world $s\in S$ is designated as the actual one, a sound a complete data stream for $s$ will, in finite time, enumerate enough information to for the $Cond$ method to delete all the other worlds, and so the actual world remains as the only one, and so also the minimal (most plausible) one. 

To see that $Cond_{CB}$ will not be able to identify this space, let us assume that for all $x\in\mathcal{P}(S)$, $D(x)=2$. We need to show that for any plausibility preorder on $S$ there is a world $s\in S$, and a sound and complete stream $\vec{O}$ for $s$, such that $Cond_{CB}$ fails to identify $s$ on $\vec{O}$. Take a preorder $\preceq$ on $S$, there are two cases, either (a) there is a unique minimal element $s$, or (b) there is none. For (a), take a $t\in S$, such that $s\preceq t$. There is a sound and complete stream $\vec{O}$ for $t$, that enumerates each observable true in $t$ exactly once. While reading that sequence, $Cond_{CB}$ will not apply a single update, and so on a sound and complete sequence for $t$ it will converge to $s$, which means it fails to identify $t$. For (b), a similar argument holds---for all among the minimal equiplausible worlds there will be a sound and complete sequence that enumerates every piece of data exactly once. On such a stream the update of $Cond_{CB}$ will not fire at all, and so there will be always more than one candidate for the actual world, so $Cond_{CB}$ will not converge to the singleton of the actual world. 

It remains to be argued that $Cond$ can identify in the limit everything that $Cond_{CB}$ can. Take an epistemic space $\mathbb{S}=(S,\mathcal{O})$, and assume that an $s \in S$ is identified in the limit by $Cond_{CB}$ on a stream $\vec{O}$ (that is sound and complete for $s$). That means that there is a $k\in \mathbb{N}$, such that for all $n\geq k$,  $L^\preceq_{Cond_{CB}}(\mathbb{S},\vec{O}[n])=\{s\}$. So, for all $t\in S$ such that $t\neq s$, $\vec{O}[n]$ includes $O \in \mathcal{O}$, such that $t\notin O$. Hence, $L^\preceq_{Cond}(\mathbb{S},\vec{O}[k])=\{s\}$, and, since $Cond$ only removes worlds, and $\vec{O}$ never enumerates anything false in $s$, $L^\preceq_{Cond}(\mathbb{S},\vec{O}[n])=\{s\}$, for all $n\geq k$. 

A similar argument works for the $Lex_{CB}$ and $Mini_{CB}$ method.
\end{proof}

Putting together Theorem \ref{th:br_uni} and Proposition \ref{prop:CB_truthtracking} we get the following corollary.

\begin{corollary}
  $Cond_{CB}$ and $Lex_{CB}$ are not universal.  
\end{corollary}

Clearly, confirmation bias can be detrimental to truth-tracking. The negative effect of stubbornness in revision can be uniformly overcome by the use of so-called fat streams, i.e., sound and complete streams that enumerate every information infinitely many times (which is possible as long as the set $\mathcal{O}$ is at most countable). Fat streams were introduced and studied before in computational learning theory in the context of memory-limited learners (see, e.g., \cite{Carlucci:2007aa}).

\paragraph{Simulation results} We ran a comparative simulation study of confirmation-biased revision and the regular unbiased revision, following the method described in Section \ref{sec:simulating_br}. The stubbornness values were randomly generated for all observables in the epistemic space as integers from $1$ to $5$.  Figure \ref{fig:cb_against_unbiased} shows the respective frequencies of truth-tracking success.

\begin{figure}[H]
    \centering
    \includegraphics[scale=0.8]{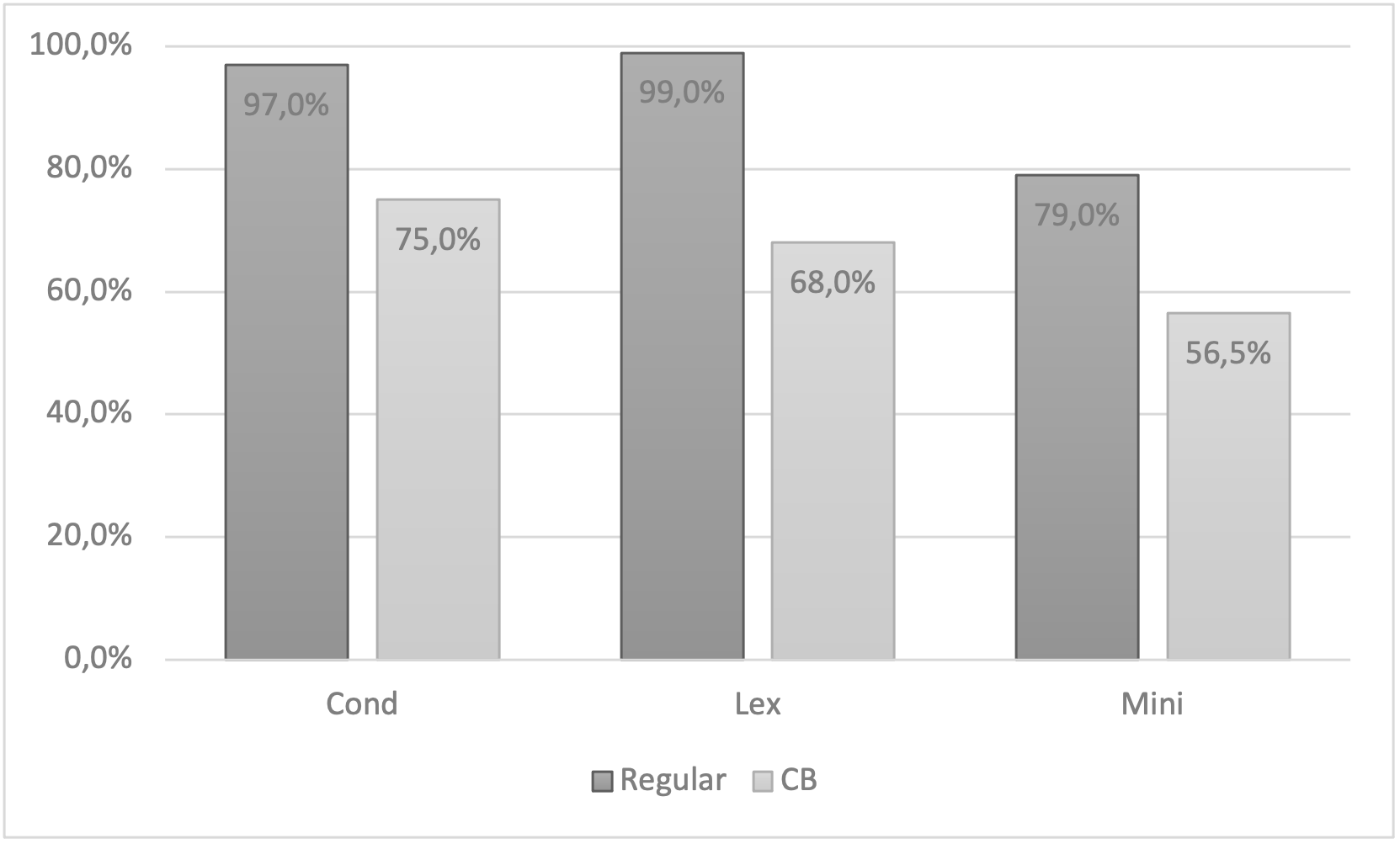}
    \caption{Confirmation-biased belief revision methods against unbiased belief revision methods}
    \label{fig:cb_against_unbiased}
\end{figure}

\subsection{Framing Bias}

Framing bias, also known as framing effect \cite{Kahneman:1979aa} refers to the fact that the way information is perceived (framed) by an agent can affect decision-making. We will introduce the framing function, $FR$ which, broadly speaking, gives a range of interpretation for an observation, i.e., the incoming information can be `re-framed' into another information, within the range allowed by $FR$.
\begin{definition}\label{framing-function}
    Given an epistemic space $\mathbb{S}=(S,\mathcal{O})$, the \emph{framing function} is  $FR:\mathcal{O}\rightarrow \mathcal{P}(S)$. 
\end{definition}
\noindent Note that the above definition is very general---we do not assume that the agent takes into account their observational apparatus, and so we allow for the observation to be interpreted as any proposition. While confirmation bias pertained to frequency of information in a stream, framing bias is related to its correctness and precision. We can pose a variety of constraints on framing, for instance we could require that the framed information is in some way related to the original information. In particular, in this paper we impose that, with the actual information $O$, the agent perceives $X$ such that $X\subseteq O$. In this case, i.e, $FR(O)\subseteq \mathcal{P}(O)$. This particular kind of framing can be seen as overconfidence bias, since given an observation with some uncertainty range, the learner sees it as one with a narrower range, i.e., one that is more certain.

As before, we will formally model the three belief revision methods, conditioning, lexicographic revision, and minimal revision under the conditions of the bias.  

\begin{definition}[Framing-bias methods]\label{fram-bias-methods}
    Let $\mathbb{B}=(S,\mathcal{O},\preceq)$ be a plausibility space, $\sigma\in \mathcal{O}^{*}$ a data sequence, $p\in \mathcal{O}$ an observable, $FR$ a framing function, and and $R_{1}$ is a one-step revision method. We define a framing-biased method in the following way:
\[ R_{FR}(\mathbb{B},\lambda)=\mathbb{B},\]
   \[ R_{FR}(\mathbb{B},\sigma\cdot p)=R_{1}(R_{FR}(\mathbb{B},\sigma),x), \text{ such that } x\in FR(p).
\]
We obtain the framing-biased conditioning, lexicographic and minimal revision $Cond_{FR}$, $Lex_{FR}$, $Mini_{FR}$ by substituting $R_{1}$ in the preceding definition by $Cond_{1}, Lex_{1}$ and $Mini_{1}$, respectively.
\end{definition}

\paragraph{Truth-tracking under framing bias} As before, we will now investigate how framing bias affects truth-tracking capabilities of belief revision methods.

\begin{definition}
Given a stream $\vec{O} =(O_0, O_1, \ldots)$ and a framing function $FR$, we define a framing of $\vec{O} $ as $FR(\vec{O} )=(P_0,P_1,\ldots)$, where for each $i\in\mathbb{N}$, $P_i\in FR(O_i)$. We will call $FR(\vec{O} )$ static iff for every $i,j\in\mathbb{N}$, with $i\neq j$, if $O_{i}=O_{j}$ then $P_{i}=P_{j}$, otherwise $FR(\vec{O} )$ is \emph{dynamic}.
\end{definition}
The first observation is that there are limit cases in which framing will not restrict the learning power of any of the revision methods, for instance when framing is a static identity function, or in more complicated, lucky cases when sound and complete streams are framed into (possibly different) sound and complete streams. In general however, framing will result in a certain kind of blindness, some worlds can get overlooked during the revision process. In particular, given an observable $O$ that is true at $s$, it might be the case that $O$ will get mapped to a set $P$, such that $s\notin P$, in other words, the agent will interpret a true observation as a proposition that is false in the actual world. This would be detrimental to any revision method. Hence, we get the following propositions.

\begin{proposition}
$Cond_{FR}$ and $Lex_{FR}$ are not universal. \end{proposition}
\begin{proposition}
$Mini$ is strictly more powerful than $Mini_{FR}$.
\end{proposition}

The dynamic framing allows for \emph{fair framing} of streams, where the agents observes input `erroneously' for finitely many steps, after which it is presented a full sound and complete stream. This is a notion analogous to that of \emph{fair streams} in \cite{Baltag:2019ab}, and the following is a direct consequence of the result therein of $Lex$ being universal on fair streams.

\begin{proposition}
$Lex_{FR}$ is universal on fairly framed streams. 
\end{proposition}

\paragraph{Simulation results} As before, we ran a comparative simulation study of confirmation-biased revision and the regular unbiased revision. As before we generate a sound and complete stream, which then gets transformed into its framed version, by applying the framing function to each observation independently. By the restrictions we impose, the framing function outputs always a random subset of the original proposition, which can be the empty set.  Figure \ref{fig:fb_against_unbiased} shows the respective frequencies of truth-tracking success.

\begin{figure}[H]
    \centering
    \includegraphics[scale=0.8]{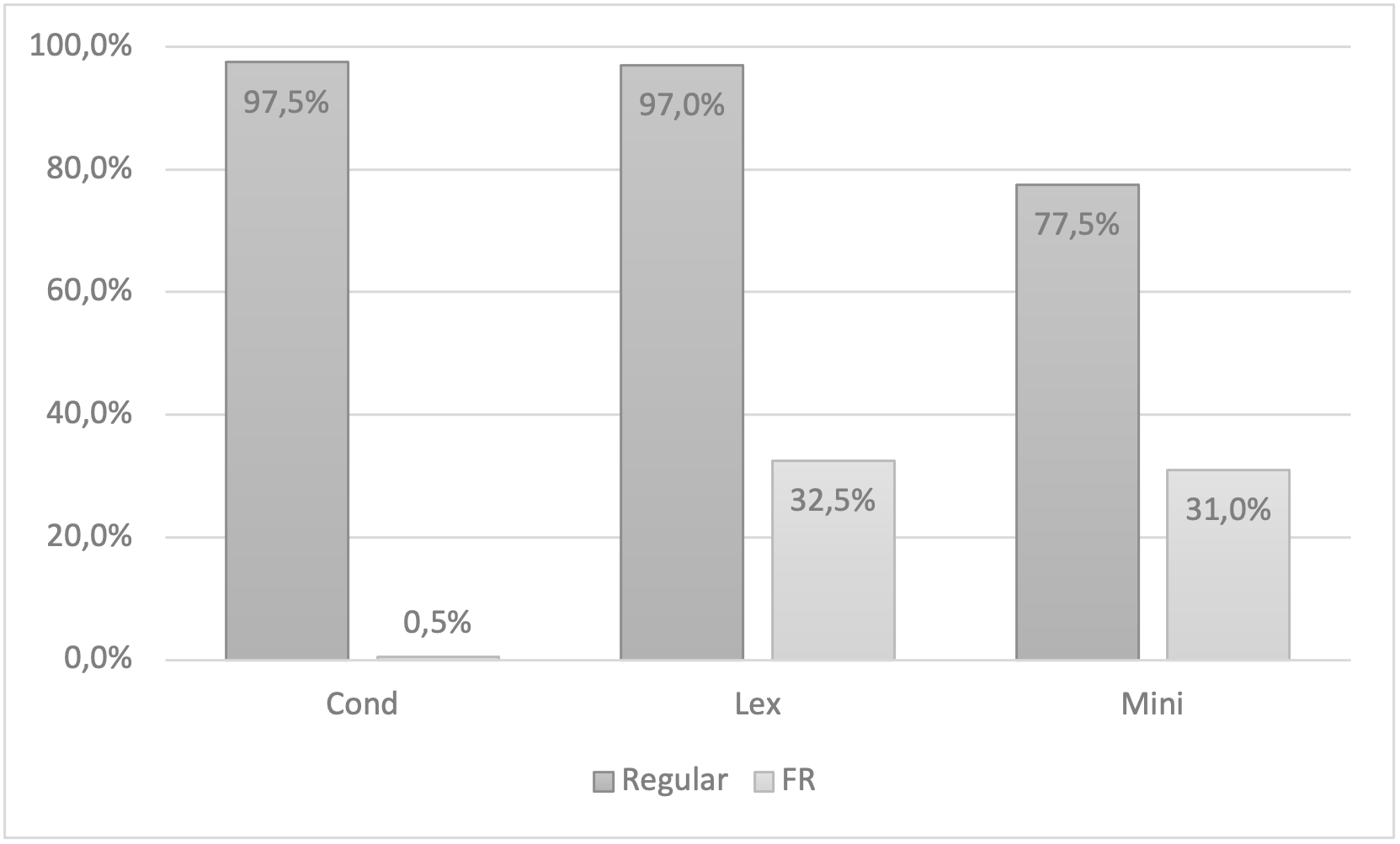}
    \caption{Framing-biased belief revision methods against unbiased belief revision methods}
    \label{fig:fb_against_unbiased}
\end{figure}

\subsection{Anchoring Bias}
Anchoring bias plays a role in decision-making influenced by the most recently received information, and it is strongly connected to lack of resources. We make everyday decisions under time pressure. These decisions are, often unconsciously, influenced by the piece of information received last before the decision point \cite{Rezaei:2021aa}. Moreover, anchoring bias in real-life scenarios can introduce a level of randomness in decision making. Consider, as an example, a student who takes part in an exam involving a multiple choice test. Due to lack of time they have to answer a question without being able to analyse it properly. While going through possible answers, the student might pick one that reminds them of something they have seen recently in their notes. 
As in the previous cases, we will provide a general definition of anchoring-biased methods. The mechanism will consists of two components, one is that the revision mechanism will always perform a minimal change, the other one is that in the case the revision step results in multiple minimal possible words, one of them will be chosen at random and made most plausible overall. In order to phrase this formally, we need several new notions. Given a set $S$, a preorder ${\preceq}\subseteq (S\times S)$, and $x\in S$, we define ${\preceq} \uparrow x:= ({\preceq} \cap (S\setminus\{x\}\times S\setminus \{x\})) \cup\{(x,s)\mid s\in S\setminus \{x\}\}$. Intuitively, this operation takes an order and outputs a new updated version of it, with $x$ upgraded to be the most plausible world. Now we will define new versions of one-step revision methods, which include in their first part the unbiased one-step revision methods and in their second part the upgrade operator.
Let $\mathbb{B}=(S,\mathcal{O},\preceq)$, $p\in \mathcal{O}$ and $Lex_1(\mathbb{B}, p)=(S,\mathcal{O},\preceq')$, we define $$Lex^+_1(\mathbb{B},p)=\begin{cases} (S,\mathcal{O},\preceq') & \text{if } |min_{\preceq'}S|=1; \\
(S,\mathcal{O},\preceq'\uparrow x), \text{ with } x\in min_{\preceq'} S & \text{otherwise}.\\
\end{cases}$$
\noindent The upgraded minimal revision, $Mini^+_1$, is defined analogously. It remains to discuss what happens when conditioning results in several minimal worlds. We propose the following interpretation.
Let $\mathbb{B}=(S,\mathcal{O},\preceq)$, $p\in \mathcal{O}$ and $Cond_1(\mathbb{B}, p)=(S',\mathcal{O}',\preceq')$, we define $$Cond^+_1(\mathbb{B},p)=\begin{cases} (S',\mathcal{O}',\preceq') & \text{if } |min_{\preceq'}S'|=1; \\
(\{x\},\mathcal{O'},\emptyset), \text{ with } x\in min_{\preceq'} S' & \text{otherwise}.\\
\end{cases}$$
\noindent $Cond^+_1$ is a very `impatient' method, as long as a singular minimal world is available, it just follows the usual drill, but if at any stage several worlds are most plausible, it picks one of them and throws away the rest of the space. This is very radical, but this way we avoid upgrading the order, which would go against the spirit of conditioning. 

\begin{definition}[Anchoring-biased methods]\label{anc-biased-methods}
    Let $\mathbb{B}=(S,\mathcal{O},\preceq)$, $\sigma\in O^{*}$ a data sequence, $p\in O$ an observable. We define the anchoring-biased methods $R_{AB}$ as:
\[R_{AB}(\mathbb{B},\lambda)=\mathbb{B},\]
\[R_{AB}(\mathbb{B},\sigma\cdot p)=R^+_{1}(R_{AB}(\mathbb{B},\sigma), min_{\preceq_{AB}}(S_{AB}\cap p)),\]
where $R_{AB}(\mathbb{B},\sigma)=(S_{AB},\mathcal{O}_{AB}, \preceq_{AB})$.
We obtain the anchoring-biased conditioning, lexicographic and minimal revision $Cond_{AB}$, $Lex_{AB}$, $Mini_{AB}$ by substituting $R^+_{1}$ above by $Cond^+_{1}, Lex^+_{1}$ and $Mini^+_{1}$, respectively.
\end{definition}

Unbiased minimal belief revision is in itself, interestingly, a form of anchoring bias. An agent using minimal belief revision actually uses the most plausible worlds where the incoming information is true to update its belief accordingly. When it comes to lexicographic revision, the definition is slightly different, but the behavior of anchoring-biased lexicographic belief revision is the same as that of unbiased minimal revision. By imposing the extra upgrade condition we  make anchoring-biased methods more `actionable', reflecting the fact that anchoring bias often plays a role in quick decision-making. After each revision step anchoring ensures that there is a candidate for the best possible world, which is randomly selected among the minimal worlds at that stage. This is especially important if resources for performing revision are limited (in the simulation these cases will be labeled `-res'). We will see that this augmentation positively affects the biased methods, even though in general the anchoring biased belief revision methods are not universal.

\paragraph{Truth-tracking under anchoring bias}\label{truth-tracking-anchoring}

Anchoring bias is most prominently connected to lack of resources. For example, when someone needs to make a decision under the pressure of time, anchoring bias can be used as heuristic. In this section we will show that, even though anchoring bias breaks universality, it can facilitate faster identification of the actual world. 

\begin{example}
Consider the plausibility space $\mathbb{B}=(\mathbb{S},\preceq)$, where $\mathbb{S}=(S,\mathcal{O})$, $S=\{w,r,s,t\}$ and $s$ the actual world. The initial plausibility order is $w\preceq t\simeq s \preceq r$, so the agent is indifferent between the worlds $t$ and $s$, and the observable propositions are $p=\{w\}, q=\{r,t,w\},\bar{p}=\{r,s,t\}$ and $\bar{q}=\{s\}$. Consider also a sound a complete data stream with respect to the actual world, $\vec{O}=(\bar{p},\ldots,\bar{q},\ldots)$. An agent using anchoring-biased conditioning identifies the actual world in the first piece of information with probability $.5$. Of course, with probability $.5$ the actual world is excluded and so the agent will not identify it. Assuming that the biased agent identifies the actual world, anchoring-biased conditioning is faster than conditioning by $k-1$ steps, where $\vec{O}_{\kappa}$ is the first occurrence of $\bar{q}$ in the data stream $\vec{O}$. Note that unbiased minimal revision will identify the world $s$ only after receiving $\bar{q}$. 

\end{example}

\noindent The above example points at the following proposition.
\begin{proposition}
$Cond_{AB}$ is not universal.
\end{proposition}
\noindent Moreover, since $Lex_{AB}$ is a version of $Mini$, based on Theorem \ref{th:br_uni}, we can state the following. 

\begin{proposition}
$Lex_{AB}$ is not universal.
\end{proposition}
\noindent Even though anchoring-biased lexicographic belief revision is not universal, it can facilitate faster truth tracking. The argument includes cases wherein the agent is indifferent between more than one most plausible worlds. Recall that an agent which uses anchoring-biased lexicographic revision revises similarly to one that uses unbiased minimal revision, but if the set of the worlds which considers most plausible is not a singleton, it selects one of the most plausible worlds with equal probability.

Unbiased minimal revision can be seen as a form of anchoring bias, as an agent that uses minimal belief revision, minimally updates its belief to be compatible with $min_{\preceq}(p)$. The difference is in the way they select the most plausible worlds after each update. Anchoring bias minimal revision and unbiased minimal revision will be compared in simulations below, where we investigate if the randomness included in anchoring-biased minimal revision improves the performance with respect to unbiased minimal revision.

\paragraph{Simulation results} We again ran a comparative simulation study of confirmation-biased revision and the regular unbiased revision, following the method described in Section \ref{sec:simulating_br}. In the case there was more than one minimal state at a certain stage of the belief revision process, the anchoring method selected one of the minimal states at random to be the conjecture of the learning method. Figure \ref{fig:cb_against_unbiased} shows the respective frequencies of truth-tracking success.
\begin{figure}[H]
    \centering
    \includegraphics[scale=0.8]{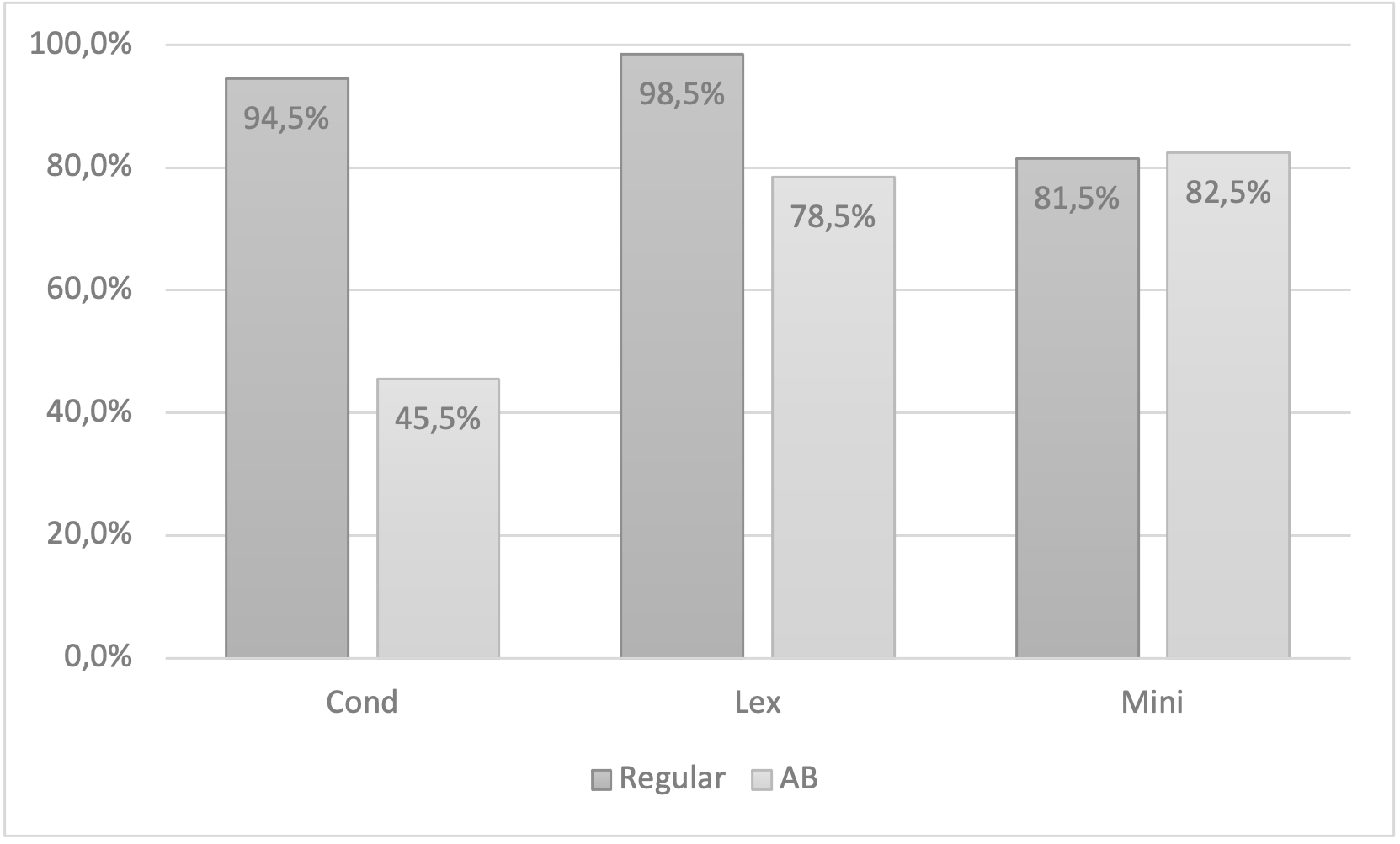}
    \caption{Anchoring-biased belief revision against unbiased belief revision}
    \label{fig:ab_against_unbiased}
\end{figure}
As anchoring bias often shows up in the context of limited resources, we run another experiment, wherein we included a parameter (a real number between $0$ and $100$) which decreases each time a revision takes place, and the process terminates when the resource is depleted. In this particular implementation, each time a revision is executed the available resource is halved and the agent stops revising when its resources fall below $1$. As we can see in Figure \ref{fig:ab_against_unbiased_resources} the anchoring ability to select a random world to be the candidate for the actual world improves the truth-tracking ability, especially in the case of minimal revision.

\begin{figure}[H]
    \centering
    \includegraphics[scale=0.8]{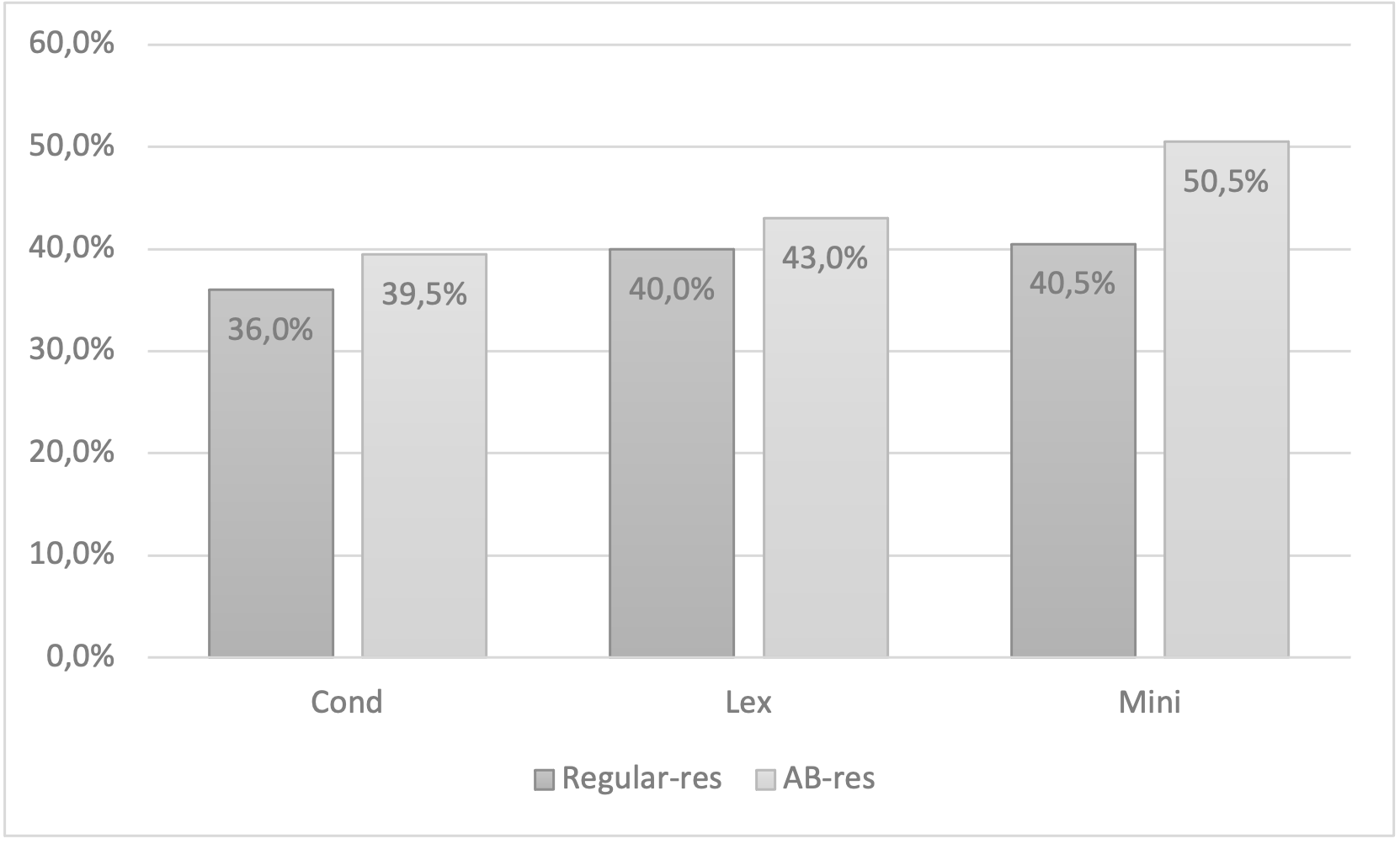}
    \caption{Anchoring-biased belief revision against unbiased belief revision - limited resources}
    \label{fig:ab_against_unbiased_resources}
\end{figure}

Finally, let us summarize some general observations about the simulation. Various components of a plausibility space affect the performance of the methods, both biased and unbiased ones. Specifically, an increase in the number of states negatively affects the performance of the belief revision methods (see Figure \ref{fig:obs_perf_anc_lex}), while an increase in the number of observables decreases the number of non-identifiable worlds, which in effect can make unbiased methods fail. More plots with the results can be found on the project repository \cite{repo}.

We also saw that, as expected, cognitive-biased belief revision methods perform worse than the unbiased ones. An exception is the anchoring-biased minimal belief revision method. Additionally, when limited resources are implemented, anchoring-biased belief revision methods perform better than the unbiased ones. This is a significant result, as it provides a potential alternative tool for truth-tracking when the resources are limited, which is usually the case in real life scenarios.

\section{Conclusions}

Cognitive bias in artificial intelligence is an interesting topic with a bright future, and as such deserves to be investigated in the context of belief revision and knowledge representation. In this paper we provided ways to formalize bias in belief-revision and learning. The three kinds of bias we discussed had completely different character, and employed different components of our belief revision based learners. We have also shown that bias can be detrimental to learning understood as truth-tracking.  

In general, biased methods are by far less reliable than the unbiased ones. 
While cognitive bias is generally problematic for truth-tracking, when resources are scarce it can be considered a tool or a heuristic. Anchoring-biased methods are a good example here, as the tests we conducted showed. This point can also serve as a rehabilitation of minimal revision, which in general is not a universal learning method.  
 
 When it comes to the simulation, we have found, in line with our expectations, that $Cond$ and $Lex$ identify the actual world in almost every test. Moreover, in general, the larger the number of observables, the higher the chances for the agent to identify the actual world. The same holds for the length of the data sequence, see Figure \ref{fig:obs_perf_anc_lex}.
 \begin{figure}[ht]
    \centering
         \includegraphics[width=.45\textwidth]{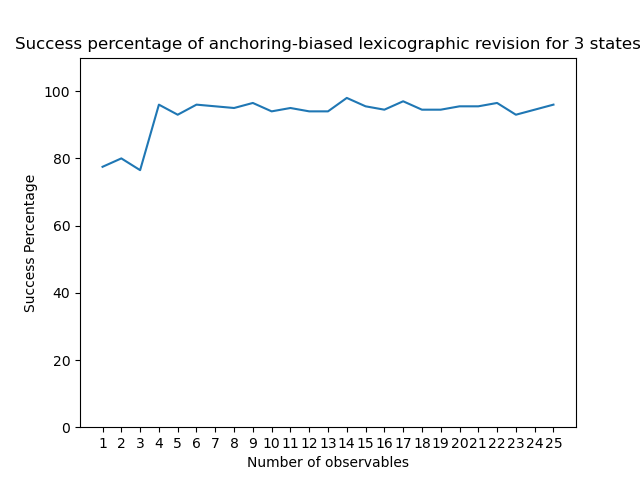}
      \includegraphics[width=.45\textwidth]{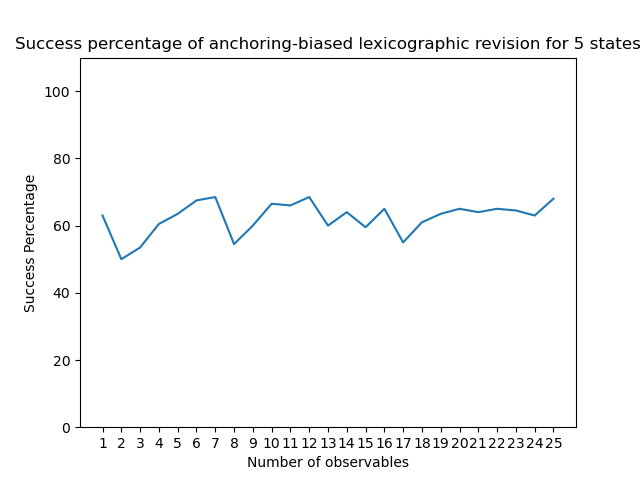}
          \includegraphics[width=.45\textwidth]{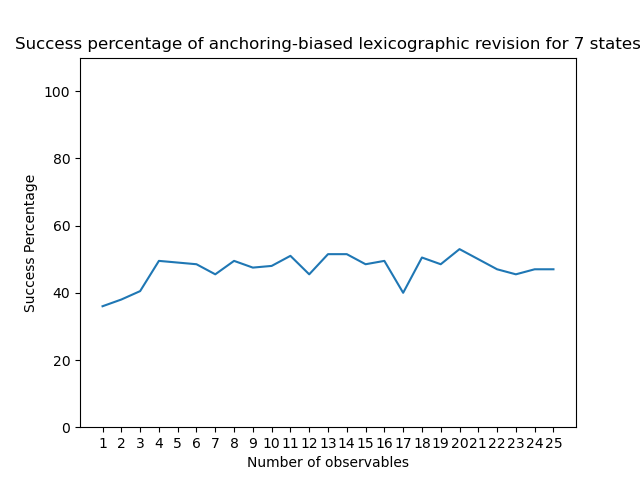}
          \includegraphics[width=.45\textwidth]{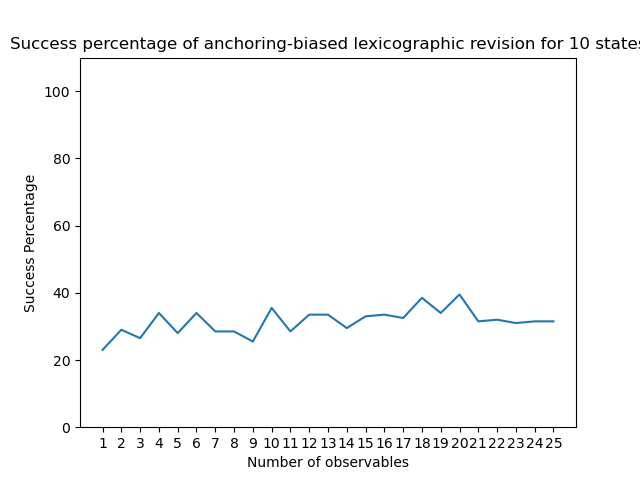}
\caption{Number of observables against anchoring bias lexicographic revision's performance for different number of possible states. The remaining graphs can be accessed in our GitHub repository \cite{repo}.}
\label{fig:obs_perf_anc_lex}
\end{figure}
\noindent Biased belief revision methods, are in general less successful than the unbiased ones---in particular, the information loss in framing can be fatal for truth-tracking by conditioning. On the other hand, anchoring bias can be used as a heuristic for faster identification.

In our work we model only some types of cognitive bias, the ones more applicable in artificial intelligence. Types mostly related to human emotional decision-making were intentionally excluded, but they would be a very interesting topic of future work. 
Moreover, although we investigated how randomness on the states, observables, and data streams affects truth-tracking, randomness of the environment itself is not a factor in this model. Assigning some bias to the elements of the tests could potentially give better insights into truth-tracking. Finally, it would be very interesting to relate our results to the existing work on resource bounded belief revision in the AGM paradigm, in particular to \cite{Wassermann:1999aa}, to look for expressibility results in the context of dynamic logic of learning theory (DLLT, \cite{dynamic-logic-for-lt}), and, last but not least, make steps towards empirical predictions for cognitive science of bias. 

\bibliographystyle{eptcsini}
\bibliography{bibliography}

\providecommand{\noopsort}[1]{}
\begin{thebibliography}{10}
\providecommand{\bibitemdeclare}[2]{}
\providecommand{\surnamestart}{}
\providecommand{\surnameend}{}
\providecommand{\urlprefix}{Available at }
\providecommand{\url}[1]{\texttt{#1}}
\providecommand{\href}[2]{\texttt{#2}}
\providecommand{\urlalt}[2]{\href{#1}{#2}}
\providecommand{\doi}[1]{doi:\urlalt{https://doi.org/#1}{#1}}
\providecommand{\eprint}[1]{arXiv:\urlalt{https://arxiv.org/abs/#1}{#1}}
\providecommand{\bibinfo}[2]{#2}

\bibitemdeclare{incollection}{BGS11}
\bibitem{BGS11}
\bibinfo{author}{A.~\surnamestart Baltag\surnameend},
  \bibinfo{author}{N.~\surnamestart Gierasimczuk\surnameend} \&
  \bibinfo{author}{S.~\surnamestart Smets\surnameend} (\bibinfo{year}{2011}):
  \emph{\bibinfo{title}{Belief revision as a truth-tracking process}}.
\newblock In \bibinfo{editor}{K.~\surnamestart Apt\surnameend}, editor:
  {\slshape \bibinfo{booktitle}{TARK'11: Proceedings of the 13th Conference on
  Theoretical Aspects of Rationality and Knowledge, Groningen, The Netherlands,
  July 12-14, 2011}}, \bibinfo{publisher}{ACM}, \bibinfo{address}{New York, NY,
  USA}, pp. \bibinfo{pages}{187--190}, \doi{10.1145/2000378.2000400}.


\bibitemdeclare{article}{Baltag:2019ab}
\bibitem{Baltag:2019ab}
\bibinfo{author}{A.~\surnamestart Baltag\surnameend},
  \bibinfo{author}{N.~\surnamestart Gierasimczuk\surnameend} \&
  \bibinfo{author}{S.~\surnamestart Smets\surnameend} (\bibinfo{year}{2019}):
  \emph{\bibinfo{title}{Truth-Tracking by Belief Revision}}.
\newblock {\slshape \bibinfo{journal}{Studia Logica}}
  \bibinfo{volume}{107}(\bibinfo{number}{5}), pp. \bibinfo{pages}{917--947},
  \doi{10.1007/s11225-018-9812-x}.


\bibitemdeclare{article}{dynamic-logic-for-lt}
\bibitem{dynamic-logic-for-lt}
\bibinfo{author}{A.~\surnamestart Baltag\surnameend},
  \bibinfo{author}{N.~\surnamestart Gierasimczuk\surnameend},
  \bibinfo{author}{A.~\surnamestart Özgün\surnameend},
  \bibinfo{author}{A.L.~\surnamestart {Vargas Sandoval}\surnameend} \&
  \bibinfo{author}{S.~\surnamestart Smets\surnameend} (\bibinfo{year}{2019}):
  \emph{\bibinfo{title}{A dynamic logic for learning theory}}.
\newblock {\slshape \bibinfo{journal}{Journal of Logical and Algebraic Methods
  in Programming}} \bibinfo{volume}{109}, p. \bibinfo{pages}{100485},
  \doi{10.1016/j.jlamp.2019.100485}.

\bibitemdeclare{article}{Benthem:2007aa}
\bibitem{Benthem:2007aa}
\bibinfo{author}{J.~\surnamestart van Benthem\surnameend}
  (\bibinfo{year}{2007}): \emph{\bibinfo{title}{Dynamic logic for belief
  revision}}.
\newblock {\slshape \bibinfo{journal}{Journal of Applied Non-Classical Logics}}
  \bibinfo{volume}{17}(\bibinfo{number}{2}), pp. \bibinfo{pages}{129--155},
  \doi{10.3166/jancl.17.129-155}.


\bibitemdeclare{inproceedings}{Boutilier:1993fj}
\bibitem{Boutilier:1993fj}
\bibinfo{author}{C.~\surnamestart Boutilier\surnameend} (\bibinfo{year}{1993}):
  \emph{\bibinfo{title}{Revision Sequences and Nested Conditionals}}.
\newblock In: {\slshape \bibinfo{booktitle}{IJCAI'93: Proceedings of the
  Thirteenth International Joint Conference on Artificial Intelligence}},
  \bibinfo{address}{Chambery, France}, pp. \bibinfo{pages}{519--525},
  \doi{10.5555/1624025.1624098}.

\bibitemdeclare{article}{Carlucci:2007aa}
\bibitem{Carlucci:2007aa}
\bibinfo{author}{L.~\surnamestart Carlucci\surnameend},
  \bibinfo{author}{J.~\surnamestart Case\surnameend},
  \bibinfo{author}{S.~\surnamestart Jain\surnameend} \&
  \bibinfo{author}{F.~\surnamestart Stephan\surnameend} (\bibinfo{year}{2007}):
  \emph{\bibinfo{title}{Results on memory-limited U-shaped learning}}.
\newblock {\slshape \bibinfo{journal}{Information and Computation}}
  \bibinfo{volume}{205}(\bibinfo{number}{10}), pp. \bibinfo{pages}{1551--1573},
  \doi{10.1016/j.ic.2007.04.001}.


\bibitemdeclare{phdthesis}{Gie10}
\bibitem{Gie10}
\bibinfo{author}{N.~\surnamestart Gierasimczuk\surnameend}
  (\bibinfo{year}{2010}): \emph{\bibinfo{title}{Knowing One's Limits. Logical
  Analysis of Inductive Inference}}.
\newblock Ph.D. thesis, \bibinfo{school}{Universiteit van Amsterdam, The
  Netherlands}.

\bibitemdeclare{incollection}{Gierasimczuk:2013aa}
\bibitem{Gierasimczuk:2013aa}
\bibinfo{author}{N.~\surnamestart Gierasimczuk\surnameend},
  \bibinfo{author}{V.F.~\surnamestart Hendricks\surnameend} \&
  \bibinfo{author}{D.~\surnamestart de~Jongh\surnameend}
  (\bibinfo{year}{2014}): \emph{\bibinfo{title}{Logic and Learning}}.
\newblock In \bibinfo{editor}{A.~\surnamestart Baltag\surnameend} \&
  \bibinfo{editor}{S.~\surnamestart Smets\surnameend}, editors: {\slshape
  \bibinfo{booktitle}{Johan van Benthem on Logic and Information Dynamics}},
  \bibinfo{publisher}{Springer International Publishing},
  \bibinfo{address}{Cham}, pp. \bibinfo{pages}{267--288},
  \doi{10.1007/978-3-319-06025-5_10}.


\bibitemdeclare{article}{Gold67}
\bibitem{Gold67}
\bibinfo{author}{E.M.~\surnamestart Gold\surnameend} (\bibinfo{year}{1967}):
  \emph{\bibinfo{title}{Language Identification in the Limit}}.
\newblock {\slshape \bibinfo{journal}{Information and Control}}
  \bibinfo{volume}{10}, pp. \bibinfo{pages}{447--474},
  \doi{10.1016/S0019-9958(67)91165-5}.

\bibitemdeclare{incollection}{Hahn:2014aa}
\bibitem{Hahn:2014aa}
\bibinfo{author}{U.~\surnamestart Hahn\surnameend} \&
  \bibinfo{author}{A.J.~\surnamestart Harris\surnameend}
  (\bibinfo{year}{2014}): \emph{\bibinfo{title}{Chapter Two - What Does It Mean
  to be Biased: Motivated Reasoning and Rationality}}.
\newblock In \bibinfo{editor}{B.H.~\surnamestart Ross\surnameend}, editor:
  {\slshape \bibinfo{booktitle}{Psychology of Learning and Motivation}},
  \bibinfo{volume}{61}, \bibinfo{publisher}{Academic Press}, pp.
  \bibinfo{pages}{41--102},
  \doi{10.1016/B978-0-12-800283-4.00002-2}.

\bibitemdeclare{book}{JORS99}
\bibitem{JORS99}
\bibinfo{author}{S.~\surnamestart Jain\surnameend},
  \bibinfo{author}{D.~\surnamestart Osherson\surnameend},
  \bibinfo{author}{J.S.~\surnamestart Royer\surnameend} \&
  \bibinfo{author}{A.~\surnamestart Sharma\surnameend} (\bibinfo{year}{1999}):
  \emph{\bibinfo{title}{Systems that Learn}}.
\newblock \bibinfo{publisher}{MIT Press}, \bibinfo{address}{Chicago}.

\bibitemdeclare{article}{Kahneman:1979aa}
\bibitem{Kahneman:1979aa}
\bibinfo{author}{D.~\surnamestart Kahneman\surnameend} \&
  \bibinfo{author}{A.~\surnamestart Tversky\surnameend} (\bibinfo{year}{1979}):
  \emph{\bibinfo{title}{Prospect Theory: An Analysis of Decision under Risk}}.
\newblock {\slshape \bibinfo{journal}{Econometrica}}
  \bibinfo{volume}{47}(\bibinfo{number}{2}), pp. \bibinfo{pages}{263--291},
  \doi{10.2307/1914185}.


\bibitemdeclare{inproceedings}{Kel98}
\bibitem{Kel98}
\bibinfo{author}{K.T.~\surnamestart Kelly\surnameend} (\bibinfo{year}{1998}):
  \emph{\bibinfo{title}{The learning power of belief revision}}.
\newblock In: {\slshape \bibinfo{booktitle}{TARK'98: Proceedings of the 7th
  Conference on Theoretical Aspects of Rationality and Knowledge}},
  \bibinfo{publisher}{Morgan Kaufmann Publishers Inc.}, \bibinfo{address}{San
  Francisco, CA, USA}, pp. \bibinfo{pages}{111--124},
  \doi{10.5555/645876.671884}.

\bibitemdeclare{article}{Kelly:1999aa}
\bibitem{Kelly:1999aa}
\bibinfo{author}{K.T.~\surnamestart Kelly\surnameend} (\bibinfo{year}{1999}):
  \emph{\bibinfo{title}{Iterated Belief Revision, Reliability, and Inductive
  Amnesia}}.
\newblock {\slshape \bibinfo{journal}{Erkenntnis}}
  \bibinfo{volume}{50}(\bibinfo{number}{1}), pp. \bibinfo{pages}{7--53},
  \doi{10.1023/A:1005444112348}.


\bibitemdeclare{}{repo}
\bibitem{repo}
\bibinfo{author}{P.~\surnamestart Papadamos\surnameend} \&
  \bibinfo{author}{N.~\surnamestart Gierasimczuk\surnameend}
  (\bibinfo{year}{2023}): \emph{\bibinfo{title}{Source Code for Simulations in
  Cognitive Bias and Belief Revision}}.
\newblock \urlprefix\url{https://github.com/papos8/BeliefRevisionSimulation}.

\bibitemdeclare{article}{Rezaei:2021aa}
\bibitem{Rezaei:2021aa}
\bibinfo{author}{J.~\surnamestart Rezaei\surnameend} (\bibinfo{year}{2021}):
  \emph{\bibinfo{title}{Anchoring bias in eliciting attribute weights and
  values in multi-attribute decision-making}}.
\newblock {\slshape \bibinfo{journal}{Journal of Decision Systems}}
  \bibinfo{volume}{30}(\bibinfo{number}{1}), pp. \bibinfo{pages}{72--96},
  \doi{10.1080/12460125.2020.1840705}.

\bibitemdeclare{article}{Rott:1989ab}
\bibitem{Rott:1989ab}
\bibinfo{author}{H.~\surnamestart Rott\surnameend} (\bibinfo{year}{1989}):
  \emph{\bibinfo{title}{Conditionals and theory change: Revisions, expansions,
  and additions}}.
\newblock {\slshape \bibinfo{journal}{Synthese}}
  \bibinfo{volume}{81}(\bibinfo{number}{1}), pp. \bibinfo{pages}{91--113},
  \doi{10.1007/BF00869346}.


\bibitemdeclare{article}{fast-and-slow-thinking}
\bibitem{fast-and-slow-thinking}
\bibinfo{author}{A.~\surnamestart Solaki\surnameend},
  \bibinfo{author}{F.~\surnamestart Berto\surnameend} \&
  \bibinfo{author}{S.~\surnamestart Smets\surnameend} (\bibinfo{year}{2021}):
  \emph{\bibinfo{title}{The Logic of Fast and Slow Thinking}}.
\newblock {\slshape \bibinfo{journal}{Erkenntnis}}
  \bibinfo{volume}{86}(\bibinfo{number}{3}), pp. \bibinfo{pages}{733--762},
  \doi{10.1007/s10670-019-00128-z}.

\bibitemdeclare{incollection}{Spohn:1988aa}
\bibitem{Spohn:1988aa}
\bibinfo{author}{W.~\surnamestart Spohn\surnameend} (\bibinfo{year}{1988}):
  \emph{\bibinfo{title}{Ordinal Conditional Functions: A Dynamic Theory of
  Epistemic States}}.
\newblock In \bibinfo{editor}{W.L.~\surnamestart Harper\surnameend} \&
  \bibinfo{editor}{B.~\surnamestart Skyrms\surnameend}, editors: {\slshape
  \bibinfo{booktitle}{Causation in Decision, Belief Change, and Statistics:
  Proceedings of the Irvine Conference on Probability and Causation}},
  \bibinfo{publisher}{Springer Netherlands}, \bibinfo{address}{Dordrecht}, pp.
  \bibinfo{pages}{105--134}, \doi{10.1007/978-94-009-2865-7_6}.


\bibitemdeclare{article}{Wassermann:1999aa}
\bibitem{Wassermann:1999aa}
\bibinfo{author}{R.~\surnamestart Wassermann\surnameend}
  (\bibinfo{year}{1999}): \emph{\bibinfo{title}{Resource Bounded Belief
  Revision}}.
\newblock {\slshape \bibinfo{journal}{Erkenntnis}}
  \bibinfo{volume}{50}(\bibinfo{number}{2}), pp. \bibinfo{pages}{429--446},
  \doi{10.1023/A:1005565603303}.


\end{thebibliography}

\end{document}